\documentclass[11pt,a4paper]{article}
\usepackage{jheppub}
\usepackage[utf8x]{inputenc}
\usepackage{amsmath}
\usepackage{amssymb}
\usepackage{float}
\usepackage{array}
\usepackage{mathrsfs} 
\usepackage{multirow}

\usepackage{amsthm}
\newtheorem{theorem}{Theorem}[section]

\newtheorem{lemma}[theorem]{Lemma}

\usepackage{hyperref}
\numberwithin{equation}{section}

\title{\boldmath The Arithmetic of Supersymmetric Vacua }
 \author{Antoine Bourget}
 \author{and Jan Troost}
 \affiliation{
Laboratoire de Physique Théorique de l'\'Ecole Normale Supérieure \\ 
CNRS,
PSL Research University,
Sorbonne Universit\'es,
75005 Paris, France}
 
\emailAdd{bourget@lpt.ens.fr}
\emailAdd{troost@lpt.ens.fr}

\abstract{We provide explicit formulas for the number of vacua of four-dimensional pure ${\cal N}=1$ super Yang-Mills
theories on a circle, with any simple gauge algebra and any choice of center and spectrum of line operators.
The formula for the $(\mathrm{SU}(N)/\mathbb{Z}_m)_n$ theory  is a key ingredient in the semi-classical calculation
of the number of massive vacua of  ${\cal N}=1^\ast$ gauge theories with gauge algebra $\mathfrak{su}(n)$,
compactified on a circle. Using arithmetic,
we express that number in an $\mathrm{SL}(2,\mathbb{Z})$ duality invariant manner. We confirm our tally of
massive vacua of the ${\cal N}=1^\ast$ theories by a count of inequivalent extrema of the exact superpotential. Furthermore, we compute a formula for a refined index that distinguishes
massive vacua according to their unbroken discrete gauge group.}

\begin{document} 
\maketitle
\flushbottom

\section{Introduction}
The calculation of the number of vacua of pure ${\cal N}=1$ super Yang-Mills theories on $\mathbb{R}^4$
is non-trivial, as witnessed by
the seventeen years it took to spell out its subtleties \cite{Witten:1982df,Witten:1997bs,Keurentjes:1998uu,Keurentjes:1999qf,Keurentjes:1999mv,Kac:1999gw,BFM,Witten:2000nv}.
The final result for a theory with simple gauge algebra $\mathfrak{g}$
is that the theory has a number of massive vacua equal to the dual Coxeter number of the gauge algebra, as predicted by chiral
symmetry breaking. Only recently it was made manifest that when we compactify the theory on a circle, further subtleties
need to be taken into account to provide the tally of massive vacua \cite{Aharony:2013hda}. In particular, the choice
of the center of the simple gauge group as well
as the spectrum of line operators in the theory influence the supersymmetric index of pure ${\cal N} =1 $ super Yang-Mills theory on a circle \cite{Aharony:2013hda}. The reasoning of
how to calculate the number of vacua in all cases was laid out in \cite{Aharony:2013hda,Aharony:2013kma}, and an explicit formula
was given for almost all cases.\footnote{See \cite{Anber:2015wha} for further analysis of the dynamics in the
non-supersymmetric setting.}

In this paper, we firstly complete the list of explicit formulas for the supersymmetric index
of the pure ${\cal N}=1$ theory on a circle. In the bulk of the paper, we calculate the number of vacua for
the $\mathrm{SU}(N)/\mathbb{Z}_m$ gauge theory, with a spectrum of line operators further specified by an integer $n$ modulo $m$
\cite{Aharony:2013hda}. 
The index is a sum over greatest common divisors. All other cases that were left aside in \cite{Aharony:2013hda} are
rendered explicit in appendix \ref{indices}.

Secondly, we compute the number of vacua of ${\cal N}=4$ supersymmetric Yang-Mills theories with gauge algebra $\mathfrak{su}(n)$, deformed by three mass terms
for the adjoint chiral multiplets to the ${\cal N}=1^\ast$ theory.  We count the 
massive vacua upon circle compactification semi-classically, by classifying gauge group breaking patterns
\cite{Donagi:1995cf,Dorey:1999sj}, including discrete gauge group factors, and using the index for the pure supersymmetric Yang-Mills theory 
on $\mathbb{R}^3 \times S^1$. The final result is an intricate combination of elementary arithmetic functions. 

Thirdly, we show that the semi-classical counting function can be rewritten, with some arithmetic effort, in terms of an $\mathrm{SL}(2,\mathbb{Z})$ invariant expression, thus proving the consistency of the semi-classical analysis with the 
duality invariance of the parent ${\cal N}=4$ theory. We confirm our census by an analysis
of inequivalent extrema of the exact elliptic superpotential \cite{Dorey:1999sj}. Finally,
we compute a refined index, that classifies massive vacua according to their unbroken discrete
gauge group, and observe that the refined index is duality invariant as well.

\section{\texorpdfstring{The Index of Pure Supersymmetric Yang-Mills Theories on a Circle}{}}
\label{pure}
Pure $\mathcal{N}=1$ theories in four dimensions with a simple gauge algebra $\mathfrak{g}$ 
permit  a choice of global gauge group
and spectrum of line operators \cite{Aharony:2013hda,Aharony:2013kma}. Our goal in this section
and appendix \ref{indices}
is to complete the list of explicit formulas for the supersymmetric index of all ${\cal N}=1$ theories,
following the path laid out in \cite{Aharony:2013hda}.

The first remaining case is the theory with (electric) gauge group $G=\mathrm{SU}(N)/\mathbb{Z}_m$, and a spectrum
of line operators with charges $(m,0) \mathbb{Z} + (n,\frac{N}{m}) \mathbb{Z}$, where both 
electric and magnetic entries are defined modulo $N$. The charges are representations of the electric times magnetic center groups $\mathbb{Z}_N \times
\mathbb{Z}_N$, and classify the possible physical behaviours of loop operators.\footnote{We refer to \cite{Aharony:2013hda} for a  pedagogical introduction of the relevant
concepts.} The resulting theory is referred to as the $(\mathrm{SU}(N)/\mathbb{Z}_m)_n$ theory. The number $m$ is a divisor of $N$,
and fixes the global choice of gauge group (via a choice of center). The integer $n$ is defined
modulo $m$ and captures a choice of complete set of  line operators. This case will play into the rest of the paper,
and we therefore treat it in the main text. All remaining cases are discussed in appendix \ref{indices}.

We calculate the number of vacua of the $(\mathrm{SU}(N)/\mathbb{Z}_m)_n$ theory upon compactification on a circle following the reasoning
described in \cite{Aharony:2013hda}. Firstly, we note that if we consider 
any pure ${\cal N}=1$, $\mathfrak{su}(n)$ theory on $\mathbb{R}^4$, then 
it has $N$ vacua. 
In the $(\mathrm{SU}(N)/\mathbb{Z}_m)_n$ theory, there is a global
symmetry group $\mathbb{Z}_{\frac{N}{m}}$ that exchanges vacua. We will label $m$ inequivalent vacua by $l=1,\dots,m$, in such a way that condensed particles in vacuum $l$ have a dyonic tilt equal to $-l$ at $\theta$-angle equal to zero.

Secondly, we consider the vacuum $l$.
It is convenient to map the dyons into purely magnetic objects 
by shifting the $\theta$-angle by $2 \pi l$.
After this transformation, the electric charge of each line operator is shifted
by $l$ times its magnetic charge, as a consequence of the Witten effect. Thus, after
this shift, we are left with a spectrum of line operator charges which is equal to
\begin{equation}
    \left(qm + p \left( n+ l \frac{N}{m} \right), p \frac{N}{m} \right)  
\end{equation}
with $q,p$ arbitrary integers, and purely magnetic condensates.
We conclude that for non-zero electric charge $qm+p(n+l \frac{N}{m})$ modulo $N$, the line operator exhibits an area
law, while if this quantity is zero (modulo $N$), the line operator satisfies a perimeter law.
A non-trivial purely magnetic line operator with a perimeter law indicates the existence of an unbroken
magnetic gauge symmetry. 

To identify the unbroken symmetry group, we determine the purely magnetic
line operator with the smallest charge. This line operator generates
the algebra of purely magnetic line operators. Its charge can be written as $\left(0 , p \frac{N}{m}\right)$ where $p$ is the smallest positive integer such that 
\begin{equation}
    \exists q \in \mathbb{Z} \, : \, qm+p\left(n+l\frac{N}{m} \right)=0 \quad \textrm{modulo }  N \, . 
\end{equation}
Since $m|N$, this implies that $p(n+l\frac{N}{m})=0$ modulo $m$. Let's introduce $r=\textrm{gcd}(m,n+l\frac{N}{m})$. We can divide the equation by 
$r$ and then infer that the solution for $p$ is $0$ modulo $\frac{m}{r}$. Thus, we finally have magnetic
line operators with charges $ (0, \frac{N}{r}) \mathbb{Z}$. 
These 't Hooft line operators have perimeter law. This shows the existence of an unbroken magnetic gauge group
$\mathbb{Z}_r \subset \mathbb{Z}_m$ in the $l$-th vacuum. 
Thus, each vacuum $l$, upon compactification on a circle, obtains a multiplicity equal to $r$. 
We conclude that the supersymmetric index $I$ counting massive vacua of the $(\mathrm{SU}(N)/\mathbb{Z}_m)_n$ pure ${\cal N}=1$ 
 theory is equal to
\begin{eqnarray}
\label{indexPure}
I^{{\cal N}=1}(N,m,n) &=& \frac{N}{m} \sum_{l=1}^{m} \textrm{gcd} \left(m,n+   \frac{lN}{m}\right) \, .
\end{eqnarray}
Thus, we have accomplished our first task of providing an explicit formula for the supersymmetric index
for all choices of $(N,m,n)$. We refer to appendix \ref{indices} for a table of results for all simple Lie algebras.

\section{\texorpdfstring{The Index of $\mathcal{N}=1^\ast$ Gauge Theories on a Circle}{}}
\label{star}
In this section, we study ${\cal N}=4$ super Yang-Mills theories deformed by three supersymmetric mass terms for the three
${\cal N}=1$ chiral multiplets in the adjoint. We wish to compute the supersymmetric index for these
${\cal N}=1^\ast$ theories on the manifold $\mathbb{R}^3 \times S^1$ and
for an $\mathfrak{su}(n)$ gauge algebra. The two main ingredients will be 
an analysis of the pattern of unbroken non-abelian gauge groups including discrete gauge group factors, and the supersymmetric index of pure ${\cal N}=1$ that we computed in section \ref{pure}.

\subsection{The Semi-Classical Calculation of the Index}
\label{semiclassical}

The number of massive vacua of the ${\cal N}=1^\ast$ theory on $\mathbb{R}^4$ with gauge algebra $\mathfrak{su}(N)$ is equal to $\sigma_1(N)$, the sum of the divisors of $N$ \cite{Donagi:1995cf}. We wish to generalize this counting function to the number of massive vacua for all ${\cal N}=1^\ast$ theories with $\mathfrak{su}(N)$ algebra compactified on a circle.
In this section, we count the vacua using semi-classical techniques. 

We briefly review the analysis of the classical vacua of the theory \cite{Donagi:1995cf,Dorey:1999sj,Bourget:2015lua}.
The vacuum expectation values of the three adjoint scalar fields are in one-to-one correspondence with the
 nilpotent orbits of $\mathfrak{su}(N)$, because they satisfy an $\mathfrak{sl}(2)$ algebra due to the   F-term equations of motions,
 and because the Jacobson-Morozov theorem 
 provides a bijection between the $\mathfrak{sl}(2)$ algebras embedded in the Lie algebra and the nilpotent orbits.
The conjugacy classes of nilpotent orbits are moreover in bijection with integer partitions of $N$. Thus, we consider 
a partition
\begin{equation}
   N = \underbrace{1+ \cdots +1}_{r_1} + \underbrace{2+ \cdots +2}_{r_2} + \cdots = \sum\limits_{i=1}^N r_i \times i \, ,
\end{equation}
where the integers $r_i$ count representations of $\mathfrak{sl}(2)$ of dimension $i$. 
\subsubsection*{The Unbroken Gauge Group}
The centralizer of the corresponding $\mathfrak{sl}(2)$-triple (which are the vacuum expectation values of the adjoint scalar
fields) in the simply connected complexified gauge group $\mathrm{SL}(n,\mathbb{C})$ is 
(see e.g. \cite{CM}):
\begin{eqnarray}
\label{centralizerSUn}
H_{\mathbb{C}} &=&   S \left( \prod\limits_{i=1}^N \left[ \mathrm{GL}(r_i) \right]^i_\Delta \right) \, . 
\end{eqnarray}
The symbol $\Delta$ is there to remind us that we diagonally embed $\mathrm{GL}(r_i)$ inside $i$ copies of the group. The symbol
$S$ indicates that the total determinant equals one.\footnote{See \cite{Bourget:2015lua} for details in the spirit of the present analysis.}
If the complexified gauge group is $G_{\mathbb{C}} = \mathrm{SL}(N,\mathbb{C})/\mathbb{Z}_m$, with $m$ a divisor of $N$, then the centralizer is obtained by dividing the group (\ref{centralizerSUn}) by $\mathbb{Z}_m$. 

We first observe that for a partition which contains two different integers, there will be an abelian gauge group factor.
The partition will only give rise to massless vacua.\footnote{It is important to note that the unbroken gauge group
does not allow for discrete Wilson lines on the circle that would break the abelian gauge group factors and leave a non-abelian unbroken gauge group after compactification. For other gauge algebras, this phenomenon does occur 
\cite{Bourget:2015cza,Bourget:2015upj}.}
Thus, to obtain massive vacua, we consider $d$ representations of dimension $\frac{N}{d}$, where $d$ divides $N$. In this configuration for the adjoint scalars,
the unbroken gauge group is 
\begin{eqnarray}
\label{residualGaugeGroup}
H_{\mathbb{C}}  = S \left(  (\mathrm{GL}(d))^{\frac{N}{d}}_\Delta \right) \cap G_{\mathbb{C}} \, . 
\end{eqnarray} 
In the low-energy limit, our initial $\mathfrak{su}(N)$, ${\cal N}=1^\ast$ theory becomes a pure $\mathcal{N}=1$ theory with  unbroken gauge group equal to the compact form of the group (\ref{residualGaugeGroup}). We wish to  understand a few properties of the unbroken group (\ref{residualGaugeGroup}).

The non-abelian gauge algebra will be $\mathfrak{sl}(d)$. The covering group $\mathrm{SL}(d)$
has a $\mathbb{Z}_d$ center. We wish to analyze the influence of modding out the
$\mathrm{SU}(N)$ group by $\mathbb{Z}_m$ on the center of this non-abelian group. To do this division, we work in the subgroup
$\mathbb{Z}_{\textrm{lcm}(m,d)} \subset \mathbb{Z}_N$. After division by $\mathbb{Z}_m$, we are left with
a $\mathbb{Z}_{\textrm{lcm}(m,d)/m}$ center for the group generated by $\mathfrak{sl}(d)$. 
In other words, the center of the non-abelian
gauge group will be
$\mathbb{Z}_{d/\textrm{gcd}(m,d)}$. Furthermore, the original center of the group was $\mathbb{Z}_{N/m}$. We therefore
have another $N / \textrm{lcm}(d,m)$ inequivalent diagonal matrices in the gauge group.

A more detailed analysis of the unbroken gauge group is given in appendix \ref{isomorphism}, where we prove that $\mathrm{SL}(d)/\mathbb{Z}_{ \textrm{gcd}(d,m)}$ is a subgroup of $H_{\mathbb{C}}$ and that
\begin{equation}
\label{structureGaugeGroup}
\frac{H_{\mathbb{C}}}{\left( \frac{\mathrm{SL}(d)}{\mathbb{Z}_{ \textrm{gcd}(d,m)}} \right)} \cong \mathbb{Z}_{\frac{N}{\textrm{lcm}(d,m)}} \, . 
\end{equation}

\subsubsection*{The Index}
We have argued that the effective non-abelian dynamics in a vacuum characterized by adjoint vacuum expectation values determined by the integer $d$
is akin to that of a pure ${\cal N}=1$ super Yang-Mills theory with gauge group $\mathrm{SU}(d)/\mathbb{Z}_{\textrm{gcd}(m,d)}$, entangled with a further set of discrete diagonal gauge group elements. The fact that the gauge group has a non-trivial component group
renders a more explicit description challenging.
We further intuit  that the 
number $n$ determining the dyonic tilt of the spectrum of line operators is inherited by the
representative $\mathrm{SU}(d)/\mathbb{Z}_{\textrm{gcd}(m,d)}$ pure
${\cal N}=1$ theory.

We then count the total number of vacua for the ${\cal N}=1^\ast$ theory as follows.
For each divisor $d$, the ${N / \textrm{lcm}(d,m)}$ discrete group elements %of the electric gauge group  
in (\ref{structureGaugeGroup})
allow us to turn on an equal number of electric Wilson lines that provide a
multiplicity to the index of a
pure $(\mathrm{SU}(d)/\mathbb{Z}_{\textrm{gcd}(m,d)})_n$ theory.
Thus, we propose to multiply the multiplicity by the index of the pure $\mathcal{N}=1$ theory (\ref{indexPure}) to compose
the index for the ${\cal N}=1^\ast$ $(\mathrm{SU}(N)/\mathbb{Z}_m)_n$ theory: 
\begin{eqnarray}
I^{{\cal N}=1^\ast}(N,m,n) 
&=& \sum_{d|N} \frac{N}{\textrm{lcm}(m,d)} I^{{\cal N}=1}(d,\textrm{gcd}(m,d),n)
\\
                          &=& \sum_{d|N} \frac{N}{\textrm{lcm}(m,d)} \times \frac{d}{\textrm{gcd}(m,d)} \sum_{l=1}^{\textrm{gcd}(m,d)}
                          \textrm{gcd}(\textrm{gcd}(m,d),n+ \frac{l d}{\textrm{gcd}(m,d)}) \, ,  \nonumber 
\end{eqnarray}
and we finally find
\begin{equation}
\label{index1star}
\boxed{ I^{{\cal N}=1^\ast}(N,m,n) = \frac{N}{m} \sum\limits_{d | N} \sum\limits_{l=1}^{\textrm{gcd} (d,m)} \textrm{gcd} \left( \textrm{gcd} (d,m) , n + \frac{l d}{\textrm{gcd} (d,m)} \right)   } \, .
\end{equation}
For pedagogical purposes, we list two special cases:
\begin{itemize}
\item For the $\mathrm{SU}(N)$ theory, 
we find the index 
\begin{equation}
\label{indexSUN}
    I^{{\cal N}=1^\ast}(N,1,0) = N \sum\limits_{d | N}  1 = N \sigma_0 (N) \, .   
\end{equation}
\item
For the $\mathrm{SU}(N)/\mathbb{Z}_N$ theories, the number of massive vacua on the circle is
\begin{equation}
\label{indexSUNZm}
    I^{{\cal N}=1^\ast}(N,N,n) = \sum\limits_{d | N} \sum\limits_{l=1}^{d} \textrm{gcd} \left(d  , l\right)    \, .
\end{equation}
\end{itemize}
In subsection \ref{duality}, we will study the duality invariance of the general index formula (\ref{index1star}),
inherited from the parent ${\cal N}=4$ theory. We have just presented two special cases (\ref{indexSUN}) and (\ref{indexSUNZm}) of the counting function
that are predicted to be the same by $S$-duality. The reader can choose to take a break in appendix \ref{AppendixBasicSDuality} where we prove the equality of these two indices. Inspired by these preliminaries, we unpack in the next subsection the surprises that the boxed arithmetic index (\ref{index1star}) has in store.

\subsection{The Arithmetic Heart}
\label{genericproof}
In this subsection, we regroup the contributions to the index (\ref{index1star}) by the divisors $d$ of $N$. In order to perform
the regrouping, we need various arithmetic identities that we will first prove using elementary number theory.

Firstly, we recall three classical identities involving the M\"obius function $\mu$ and the Euler totient function $\varphi$:
\begin{equation}
\label{classicalIdentities}
    \sum\limits_{d|n} \mu (d) = \delta [n=1] \, , 
    \qquad 
    \sum\limits_{d | n} \varphi (d) =  n
    \, ,
\qquad
    \sum\limits_{d|n} d \mu \left( \frac{n}{d}\right) = \varphi (n)  \, . 
\end{equation}
The function $\delta[n=1]$ is the unit function, which is unity when $n=1$ and zero on any integer larger than one. It is the identity element of Dirichlet convolution.\footnote{The Dirichlet convolution of two functions $f,g$ defined on the positive integers is given by the formula $(f \ast g)(n)=\sum_{d|n} f(d) g(n/d)$.} 
The first identity can be rephrased as the fact that the M\"obius function in Dirichlet convolution with the
constant function equal to one is the unit function, $\mu * 1 = \delta$. The second equality, proven in formula (\ref{sumGauss}), expresses the fact that
the Euler totient function convoluted with the constant function equals the identity function, $\varphi * 1 = \mathrm{id}$. The third identity, which can be written $\mathrm{id} * \mu = \varphi$, is a consequence of the other two.\footnote{Many elementary textbooks in arithmetic
contain these facts. The last equation for instance dates back to Gauss.}

Secondly, we will need the more advanced formula \cite{CS}
\begin{equation}
\label{sumGCDaffine}
    \sum\limits_{s=1}^n \textrm{gcd} \left( n , as + b \right) = n \sum\limits_{d | n} \frac{\varphi (d)}{d} \textrm{gcd} (a,d) \delta \left[ \textrm{gcd} (a,d) | b \right]  \, , 
\end{equation}
whose proof is recalled in appendix \ref{AppendixProofGCD}. The $\delta$-function again evaluates to one if the condition inside the bracket is satisfied, and to zero otherwise.

Thus, we are armed with the necessary tools to regroup the index for the ${\cal N}=1^\ast$ theory. We start from formula (\ref{index1star}), and recount:
\begin{eqnarray}
I^{{\cal N}=1^\ast}(N,m,n) &=& \frac{N}{m} \sum\limits_{d | N} \sum\limits_{k=1}^{\textrm{gcd} (d,m)} \textrm{gcd} \left( \textrm{gcd} (d,m) , n + \frac{k d}{\textrm{gcd} (d,m)} \right) \nonumber \\ 
&=& \frac{N}{m} \sum\limits_{d | N} \sum\limits_{e | \textrm{gcd} (d,m) } \frac{\varphi (e)}{e} \textrm{gcd} (d,m) \textrm{gcd} \left( \frac{d}{\textrm{gcd} (d,m)} , e \right) \delta \left[ \textrm{gcd} \left( \frac{d}{\textrm{gcd} (d,m)} , e \right) \lvert n \right] \nonumber \\
&=& \frac{N}{m} \sum\limits_{d | N} \sum\limits_{e | \textrm{gcd} (d,m) } \frac{\varphi (e)}{e} \textrm{gcd} (d,me)  \delta \left[ \textrm{gcd} \left( \frac{d}{\textrm{gcd} (d,m)} , e \right) \lvert n \right] \nonumber \\
&=& \frac{N}{m} \sum\limits_{e |m} \sum\limits_{e |d|N } \frac{\varphi (e)}{e} \textrm{gcd} (d,me)  \delta \left[ \textrm{gcd} \left( \frac{d}{\textrm{gcd} (d,m)} , e \right) \lvert n \right] \nonumber \\
&=& \frac{N}{m} \sum\limits_{e |m} \sum\limits_{d| \frac{N}{e} } \varphi (e) \textrm{gcd} (d,m)  \delta \left[ \textrm{gcd} \left( \frac{de}{\textrm{gcd} (de,m)} , e \right) \lvert n \right] \, . 
\end{eqnarray}
We now simplify the condition in the $\delta$-function. For any divisor $e$ of $m$ and any divisor $d$ of $\frac{N}{e}$, we have
\begin{eqnarray}
\textrm{gcd} \left( \frac{de}{\textrm{gcd} (de,m)} , e \right) | n &\Leftrightarrow& e \times  \textrm{gcd} (d,m)  | n \times \textrm{gcd} (de,m) \nonumber \\
&\Leftrightarrow& \frac{\textrm{gcd} (d,m)}{\textrm{gcd} (d, \frac{m}{e} ) } | n \nonumber \\
&\Leftrightarrow& \textrm{gcd} (d,m) | dn \quad \textrm{and} \quad \textrm{gcd} (d,m) | \frac{mn}{e} \nonumber \\
&\Leftrightarrow& e | \frac{mn}{\textrm{gcd} (d,m) } \, . 
\end{eqnarray}
We go back to the computation, and use the third identity in (\ref{classicalIdentities}): 
\begin{eqnarray}
I^{{\cal N}=1^\ast}(N,m,n) &=& \frac{N}{m} \sum\limits_{e |m} \sum\limits_{d| \frac{N}{e} } \varphi (e) \textrm{gcd} (d,m)  \delta \left[ e | \frac{mn}{\textrm{gcd} (d,m) } \right] \nonumber \\
&=& \frac{N}{m} \sum\limits_{d|N} \textrm{gcd} (d,m)  \sum\limits_{e | \textrm{gcd} (m  , \frac{N}{d} , \frac{mn}{\textrm{gcd} (d,m)}) } \varphi (e) \nonumber \\
&=& \frac{N}{m} \sum\limits_{d|N} \textrm{gcd} (d,m)  \textrm{gcd} \left(m  , \frac{N}{d} , \frac{mn}{\textrm{gcd} (d,m)} \right) \, , 
\end{eqnarray}
and the last line simplifies to our final formula:
\begin{equation}
\label{index1stargrouped}
\boxed{
    I^{{\cal N}=1^\ast}(N,m,n) = N \sum\limits_{d|N}  \textrm{gcd} \left(d,m  , \frac{N}{d} , \frac{N}{m} , n \right) } \, . 
\end{equation}
The original index formula (\ref{index1star}) provides information on how each semi-classical
pure ${\cal N}=1$ vacuum obtains a degeneracy upon compactification. The formula (\ref{index1stargrouped})
regroups all contributions by the gauge group symmetry breaking pattern only.

\subsection{Duality Invariance}
\label{duality}
With the formula (\ref{index1stargrouped}) in hand, we  check the duality invariance of the index of ${\cal N}=1^\ast$ theory
compactified on a circle. The parent ${\cal N}=4$ supersymmetric Yang-Mills theory with $\mathfrak{su}(N)$ gauge algebra takes on 
$\sigma_1(N)= \sum_{d|N} d$ different guises, depending on the choice of center and the dyonic tilt in the 
spectrum of line operators. 
The duality group that leaves invariant a given theory is $\Gamma_0(N)$, while the $\mathrm{SL}(2,\mathbb{Z})$ action
on the coupling maps a theory to an equivalent one, with another choice of global gauge group and spectrum
of line operators \cite{Aharony:2013hda}.
The ${\cal N}=1^\ast$ theory inherits these duality properties from its ${\cal N}=4$ parent. (See e.g.
\cite{Aharony:2000nt} for a detailed discussion.)

We recall the duality transformation rules from \cite{Aharony:2013hda}.
The $T$-duality maps the following theories into each other:
\begin{equation}
T :    \left( \frac{\mathrm{SU}(N)}{\mathbb{Z}_m} \right)_n \longrightarrow \left( \frac{\mathrm{SU}(N)}{\mathbb{Z}_m} \right)_{n + \frac{N}{m}}
\, . 
\end{equation}
In formula (\ref{index1star}), the shift $n \rightarrow n + \frac{N}{m}$ can be offset by a shift $l \rightarrow l - \frac{N}{\textrm{lcm} (d,m)}$. It is natural to find $T$-duality realized in the semi-classical counting because $T$-duality reshuffles the sum over pure ${\cal N}=1$ vacua in a straightforward manner. This proves that the index (\ref{index1star}) is $T$-invariant. In the guise (\ref{index1stargrouped}), with vacua grouped by divisors of $N$ only, $T$-duality is also manifest. 

On the other hand, only in the grouped formula (\ref{index1stargrouped}) is $S$-duality manifest. 
Let us first note that when $n=0$, $S$-duality boils down to the exchange $m \leftrightarrow \frac{N}{m}$, which is a symmetry of (\ref{index1stargrouped}).
More generally, $S$-duality can be seen as the symmetry exchanging $m$ and $\textrm{gcd}(m,n)$. Concretely, if two theories
\begin{equation}
   S :  \left( \frac{\mathrm{SU}(N)}{\mathbb{Z}_m} \right)_n \longleftrightarrow \left( \frac{\mathrm{SU}(N)}{\mathbb{Z}_{m'}} \right)_{n'}
\end{equation}
are $S$-dual, then we have 
\begin{equation}
m' = \frac{N}{\textrm{gcd} (m,n)} \, , \qquad \quad
\textrm{gcd} (m',n') = \frac{N}{m} \, . 
\end{equation}
Armed with this knowledge, we can write 
\begin{equation}
    I^{{\cal N}=1^\ast}(N,m,n) = N \sum\limits_{d|N}  \textrm{gcd} \left(d, \frac{N}{d} , \frac{N}{m} , \textrm{gcd} (m,n) \right)  \, , 
\end{equation}
which renders $S$-duality evident.

\subsection{The Exact Superpotential Approach}

In this subsection, we reproduce the index formula from  a description of 
the full quantum dynamics of the theory compactified on a circle. A third guise of the supersymmetric
index will be found. Our tool of choice is the effective
elliptic superpotential $\mathcal{W}$ derived in \cite{Dorey:1999sj}, which is a function of the Wilson line
and dual photon variables in the Cartan subalgebra, that are the natural variables on the Coulomb branch
upon compactification.
As discussed in detail in \cite{Bourget:2015cza,Bourget:2015upj}, the resulting elliptic Calogero-Moser
integrable potential has a periodicity equal to the fundamental weight
lattice both in the electric (Wilson line) direction (parametrized by real multiples of the complex number $\tau$ say), as well as in the magnetic (dual photon) direction (parametrized by real numbers). In the gauge theory, however, we have a larger
set of gauge inequivalent variables, and we must consider a $N$-fold cover of the fundamental torus of the integrable system, corresponding to identifications of the Cartan variables under shifts by the root lattice. Let's normalize the integrable system variables such that the fundamental torus has periodicities $1$ and $\tau$
in the dual photon and Wilson line directions respectively.
We moreover denote multiples of the fundamental torus periods by the pairs $(x,y) \in \mathbb{Z} \times \mathbb{Z}$.
From the spectrum of line operators in the theory $(\mathrm{SU}(N)/\mathbb{Z}_m)_n$, it then follows that gauge transformations
identify each individual Cartan variable under shifts in $(m,0) \mathbb{Z} + (-n,\frac{N}{m}) \mathbb{Z}$.\footnote{To prove this, it is sufficient to consider all gauge invariant Wilson-'t Hooft line operators, and perform constant shifts on the vacuum expectation values of the electric and magnetic gauge fields.}

In a second stage, we recall that all the massive vacua, when projected onto the fundamental torus
are described by index $N$ lattice extensions of the fundamental torus \cite{Donagi:1995cf,Dorey:1999sj}. These lattices are coded in
an integer $d$ and a tilt $t$, and give rise to $N$ points on the fundamental torus, which are all distinct
multiples of $(d/N,0)$ and $(t/N,1/d)$. Because we are dealing with the $\mathfrak{su}(n)$ gauge algebra,
the $N$ points of the integrable system are only determined up to an overall arbitrary shift.

We must count the solutions of the gauge theory exact superpotential (with the gauge theory notion of inequivalence),
which project onto these lattices in the fundamental torus. We illustrate the above notions concretely in the example of the gauge algebra $\mathfrak{su}(2)$, before counting inequivalent solutions in the general case.

\subsubsection*{\texorpdfstring{An Example : The $\mathfrak{su}(2)$ Gauge Algebra}{}}

For the $\mathrm{SU}(2)$ gauge group, we have one Wilson line proportional to $\tau$, and 
periodic with period $2 \tau$, arising from identifications of the Wilson line up to shifts in the root lattice. 
The dual photon has a period which is equal to the weight lattice. The complexified coordinate $Z$ in the Cartan subalgebra
thus lives in a
double cover of the torus, with periods $(1,0)$ and $(0,2 \tau)$. We can count the vacua semi-classically following the analysis
in subsection \ref{semiclassical}. We have two nilpotent orbits, corresponding to the partitions $1+1$ and $2$. The first
choice gives rise to two confining vacua of pure ${\cal N}=1$ $\mathrm{SU}(2)$, while the second yields two Higgsed vacua,
with differing discrete Wilson lines, due to an unbroken $\mathbb{Z}_2$ discrete gauge group.
The non-zero discrete Wilson line shifts the coordinate $Z$ by $\tau$.
The confining vacuum has Wilson lines which are fixed by the exact effective superpotential.
If one attempts to add a Wilson line to a confining vacuum, it is again gauge equivalent to a confining
vacuum. This describes the situation for the $\mathrm{SU}(2)$ gauge group. The vacua are illustrated in the second line of
Figure \ref{figuresu2} and Table \ref{tablesu2}.

For the $\mathrm{SO}(3)_+$ theory, we have periodicities $(2,0)$ and $(0,\tau)$ and $3+1$ semi-classical vacua, and similarly
for the $\mathrm{SO}(3)_-$ case, where the periodicities are $(2,0)$ and $(0,\tau+1)$. The magnetic $\mathbb{Z}_2$ gauge group arising in the
pure ${\cal N}=1$ $\mathrm{SO}(3)_\pm$ theories permits
a magnetic  Wilson line leading to a third non-gauge-equivalent vacuum, while the Higgs vacuum in these theories is unique.
For these two theories as well, we draw the relevant double covers, and exhibit the four inequivalent vacua
on each of these tori in Figure \ref{figuresu2}. We also summarize the data in Table \ref{tablesu2}.
These three theories are mapped into each other under $\mathrm{SL}(2,\mathbb{Z})/\Gamma_0(2)$, and in this way we verify the 
duality invariance of the index, equal to four in all theories.
\begin{table}[H]
\begin{center}
\begin{tabular}{|c|c|c|c|}
\hline
Theory & Partition $2=1+1$ & Partition $2=2$ & Total \\ \hline
$\mathrm{SU}(2)$ & $\mathrm{SU}(2) \rightarrow 1+1$ vacua & $\mathbb{Z}_2 \rightarrow 2$ vacua & 4 \\ \hline
$\mathrm{SO}(3)_+$ & $\mathrm{SO}(3)_+ \rightarrow 2+1$ vacua & $1 \rightarrow 1$ vacuum & 4 \\ \hline
$\mathrm{SO}(3)_-$ & $\mathrm{SO}(3)_- \rightarrow 1+2$ vacua & $1 \rightarrow 1$ vacuum & 4 \\ \hline
\end{tabular}
\end{center}
\caption{An example of ${\cal N}=1^\ast$ $\mathfrak{su}(2)$ vacua on the circle}
\label{tablesu2}
\end{table}

\begin{figure}
    \centering
    \includegraphics[width=.8\textwidth]{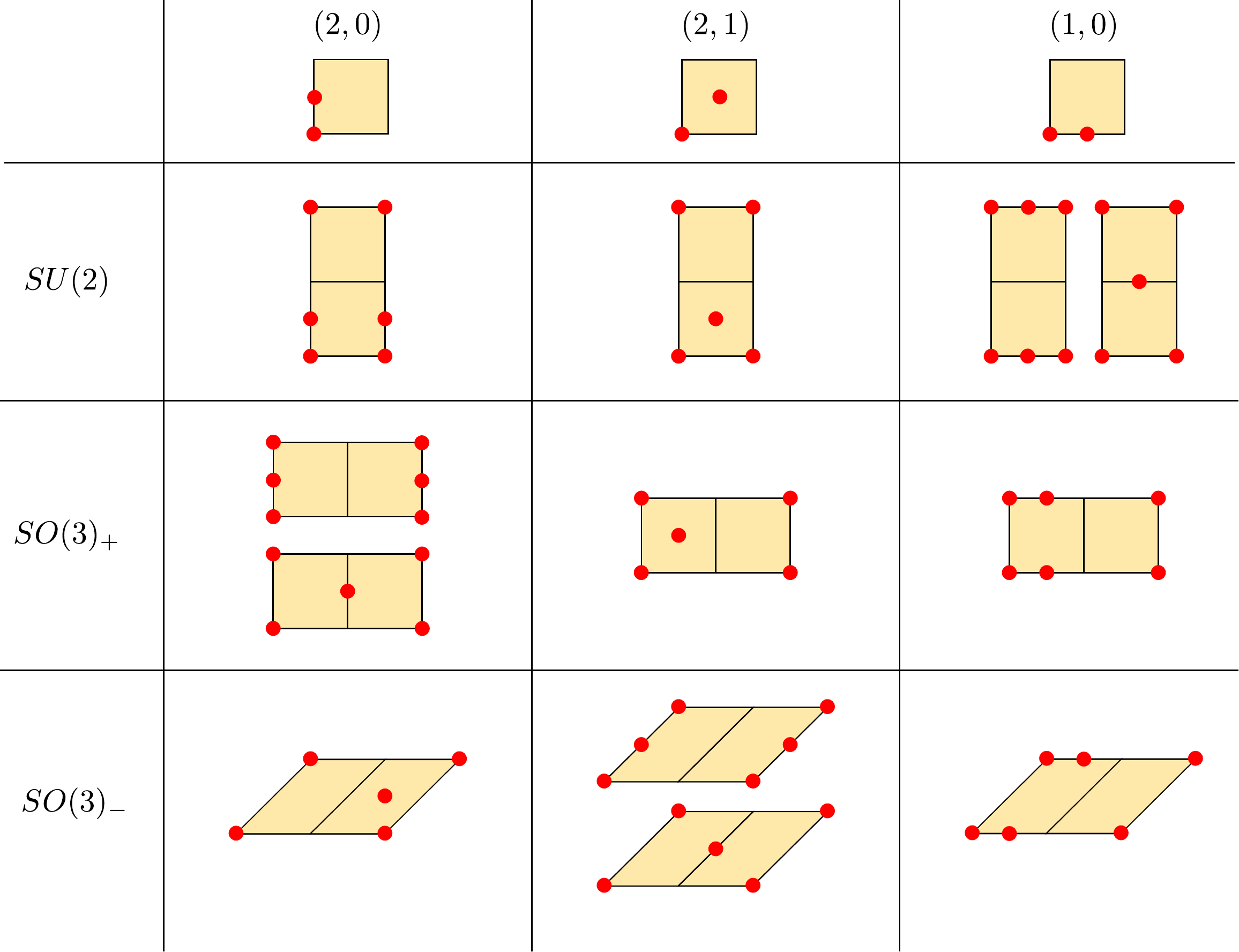}
    \caption{We draw the double cover of the torus on which the complexified gauge theory Wilson line lives.
    The double cover depends on the theory. In the resulting torus, we draw the inequivalent extrema of the 
    effective superpotential. The first two columns correspond to the confining theory, partition $1+1$, 
    and the third column corresponds to the Higgs vacua, partition $2$. The top line provides the projection
    of the extrema onto the fundamental torus. 
    }
    \label{figuresu2}
\end{figure}

\subsubsection*{The Analysis for $\mathfrak{su}(n)$}
Let's finally count the number of inequivalent massive extrema in the case of the $\mathfrak{su}(n)$ theories characterized
by the integers $(N,m,n)$.
Consider a configuration $(d,t)$ of the integrable system (i.e. projected on the fundamental torus), and lift it to the gauge theory torus parametrized by $(N,m,n)$. We want to enumerate the inequivalent configurations, under the following equivalence relation. For any $u \in \mathbb{Z}^2$, the two configurations $(X_1 , \dots , X_N)$ and $(X_1 , \dots , X_i + u , \dots , X_j -u , \dots , X_N)$ are equivalent due to the identification of gauge theory variables under
root lattice shifts. 
Furthermore, the ordering of the elements in the $N$-plet $(X_1 , \dots , X_N)$ is irrelevant because of the action of the Weyl group. 
It is clear that any configuration is equivalent to a reduced one where $X_i \in [0,1[^2$ for $1 \leq i \leq N-1$. Moreover, given a configuration characterized by the numbers $(d,t)$, one can show that 
\begin{equation}
    (X_1 , \dots , X_{N-1} , X_N) \sim (X_1 , \dots , X_{N-1} , X_N + u_1) \sim (X_1 , \dots , X_{N-1} , X_N+u_2)
\end{equation}
where $u_1 = (d,0)$ and $u_2 = (t , \frac{N}{d}) $. Because of the invariance of all line operators under
gauge transformations, we also have 
\begin{equation}
    (X_1 , \dots , X_{N-1} , X_N) \sim (X_1 , \dots , X_{N-1} , X_N + u_3) \sim (X_1 , \dots , X_{N-1} , X_N+u_4)
\end{equation}
for $u_3=(m,0)$ and $u_4 = (-n , \frac{N}{m})$. 
Hence the number of non-equivalent configurations is the index of the lattice $\mathbb{Z}^2$ (thought of
as corresponding to the last entry $X_N$ in the extremum in its reduced form) in the lattice $\bigoplus\limits_{\alpha=1}^4 u_\alpha \mathbb{Z}$. This index is given by the greatest common divisor of the determinants of all pairs of vectors $(u_\alpha , u_\beta)$ for $1 \leq \alpha < \beta \leq 4$.\footnote{\label{footnote}Indeed, 
consider a rank $n$ family $(u_\alpha)_{\alpha = 1 , \dots , k}$ of vectors in $\mathbb{Z}^n$, with $k \geq n$. Let $L$ be the sublattice of $\mathbb{Z}^n$ generated by this family. The Smith normal form of the matrix formed by the list of vectors is $\textrm{diag} (e_1 , \dots , e_{n})$ where $e_1 \times \dots \times e_i $ is the greatest common divisor of all $i \times i$ minors of the matrix, for $i \leq n$. We then have the equivalence
\begin{equation}
    \frac{\mathbb{Z}^n}{L} \cong \bigoplus\limits_{\alpha=1}^{n}  \frac{\mathbb{Z}}{e_i\mathbb{Z}} \, .
\end{equation}
The index of $\mathbb{Z}^n$ in this lattice is $\prod\limits_{i=1}^{n} e_i$, which is equal to the determinant of all $n \times n$ minors. This result also provides a basis of generators.} 
The result is 
\begin{eqnarray}
   \left[ \mathbb{Z}^2 : \bigoplus\limits_{\alpha=1}^4 u_\alpha \mathbb{Z} \right] &=& \textrm{gcd} \left( N , 0 ,  N \frac{d}{m} , N \frac{m}{d} , N \left( \frac{t}{m} + \frac{n}{d} \right) , N \right) \nonumber \\
   &=& \textrm{gcd} \left( N \frac{d}{m} , N \frac{m}{d} , N \left( \frac{t}{m} + \frac{n}{d} \right) \right) \, . 
\end{eqnarray}
We used the fact that a divisor of $N \frac{d}{m}$ and $N \frac{m}{d}$ is necessarily a divisor of $N$. Finally, we obtain a third formula for the number of vacua, based on the count of inequivalent solutions to the effective superpotential: 
\begin{equation}
\label{CountingLattices}
\boxed{
    I^{{\cal N}=1^\ast}(N,m,n) = \sum\limits_{d|N}  \sum\limits_{t=0}^{d-1} \textrm{gcd} \left( N \frac{d}{m} , N \frac{m}{d} , N \left( \frac{t}{m} + \frac{n}{d} \right) \right)  \, . } 
\end{equation}
This gives the multiplicity assigned to all sublattices of index $N$, which are labelled by $(d,t)$. In other words, if we draw Figure \ref{figuresu2} for algebra $\mathfrak{su}(N)$, in the box corresponding to column $(d,t)$ and line $\mathrm{SU}(N)/\mathbb{Z}_m$, there will be $\textrm{gcd} \left( N \frac{d}{m} , N \frac{m}{d} , N \left( \frac{t}{m} + \frac{n}{d} \right) \right)$ inequivalent lattices. 

\subsubsection*{Equivalence of the formulas}
No matter how we slice the index, the total tally should remain the same. Thus,
we must show that the index (\ref{CountingLattices}) obtained through the exact superpotential approach is identical to the count (\ref{index1star}) in the semi-classical analysis. This represents a non-trivial consistency check of the validity of our formulas. 
The proof involves more arithmetical lemmas that we relegate to appendix \ref{AppendixLemmas}. We first use the fact that 
 \begin{equation}
     \textrm{gcd} \left( N \frac{d}{m} , N \frac{m}{d} , N \left( \frac{t}{m} + \frac{n}{d} \right) \right) = \frac{N}{\textrm{lcm} (m,d)}    \textrm{gcd} \left( d,m,  \frac{td + nm}{\textrm{gcd} (d,m)}\right)
 \end{equation}
 to rewrite (\ref{CountingLattices}) as 
 \begin{equation}
    I^{{\cal N}=1^\ast}(N,m,n) = \frac{N}{m} \sum\limits_{d|N}  \sum\limits_{t=1}^{\textrm{gcd} (d,m)}   \textrm{gcd} \left( d,m,  \frac{td + nm}{\textrm{gcd} (d,m)}\right)\, . 
\end{equation}
We introduce the shorthand $\delta =\textrm{gcd} (d,m) $, $d' = \frac{d}{\delta}$ and $m' = \frac{m}{ \delta}$, so that $\textrm{gcd} (d',m')=1$. We moreover define the cyclically ordered sequence
$D_a(\delta ,d') = \left[ \mathrm{gcd}(\delta ,td'+ a) \right]_{ t \in \mathbb{Z}_{\delta } }$. 
The index (\ref{CountingLattices}) becomes 
 \begin{equation}
    I^{{\cal N}=1^\ast}(N,m,n) = \frac{N}{m} \sum\limits_{d|N}  \sum\limits_{t=1}^{\delta}   \textrm{gcd} \left( \delta ,  td' + nm' \right) 
    = \frac{N}{m} \sum\limits_{d|N}  \sum\limits_{x \in D_{nm'}(\delta , d')}x
    \, .  
\end{equation}
We must now refer to  Lemma \ref{lemma2} that shows that $(\delta , nm')$, which manifestly belongs to the sequence $D_{nm'}(\delta , d')$, is also an element of $D_{n}(\delta , d')$. From Lemma \ref{lemma3} we then deduce that these two sequences are equal. We conclude that
 \begin{equation}
    I^{{\cal N}=1^\ast}(N,m,n) = \frac{N}{m} \sum\limits_{d|N}\sum\limits_{x \in D_{n}(\delta , d')}x = \frac{N}{m} \sum\limits_{d|N}  \sum\limits_{t=1}^{\delta}   \textrm{gcd} \left( \delta ,  td' + n \right) 
    \,  , 
\end{equation}
which is exactly equation (\ref{index1star}). The arithmetic sticks.

\subsubsection*{The Refined Index}
Finally, we derive a refinement of the index $I^{{\cal N}=1^\ast}(N,m,n)$ in which we keep track of the unbroken discrete gauge group in each vacuum.\footnote{We are grateful to the anonymous JHEP referee for the suggestion to refine our index.} We perform the refinement in terms of the lattice enumeration of massive
vacua. Using the notations of footnote \ref{footnote}, we propose that the discrete gauge group in the vacua corresponding  to lattices labelled by $(d,t)$ is 
\begin{equation}
\label{discretegaugegroup}
   G^{\textrm{discrete}} = \mathbb{Z}_{e_2} \times \mathbb{Z}_{e_1} \, , 
\end{equation}
where 
\begin{eqnarray}
     e_1 &=& \textrm{gcd}\left(d,m,\frac{N}{d},\frac{N}{m},t,n \right) \\
     e_1 e_2 &=& \textrm{gcd}\left(N\frac{m}{d},N\frac{d}{m},N\left(\frac{t}{m}+\frac{n}{d} \right)\right) \nonumber \, . 
\end{eqnarray}
The direct product in the group $(\ref{discretegaugegroup})$ arises from the Smith normal form, which gives
a product group structure to the class of vacua under study.
The identification of the gauge group is done under the assumption that the full degeneracy of
the vacua corresponds to turning on Wilson lines in the unbroken
discrete gauge group.
We introduce two fugacities $x$ and $y$ for the two factors of $G^{\textrm{discrete}}$, so that the coefficient of $x^{e_1} y^{e_2}$ in the refined index $I^{{\cal N}=1^\ast}(N,m,n,x,y)$ is the number of vacua with discrete gauge group (\ref{discretegaugegroup}). The final formula reads
\begin{eqnarray}
     I^{{\cal N}=1^\ast}(N,m,n,x,y) &=& \sum\limits_{d | N} \sum\limits_{t=0}^{d-1} e_1 e_2 x^{e_1} y^{e_2} \, . 
\label{latticerefined}
\end{eqnarray}
An alternative expression for the refined index with a semi-classical flavor can be obtained from formula
(\ref{latticerefined}) following the same path that we used in the previous paragraph, and reads
\begin{equation}
\begin{split}
    I^{{\cal N}=1^\ast}(N,m,n,x,y) = \frac{N}{\textrm{lcm}(m,d)} \sum_{d|N} \sum_{l=1}^{d} & \textrm{gcd}\left( d,m,n+ \frac{ld}{\textrm{gcd}(d,m)} \right) \\ 
 & x^{\textrm{gcd}(d,l,N/d,m,n,N/m)} y^{\frac{N}{\textrm{lcm}(m,d)} \frac{\textrm{gcd}(d,m,n+ld/\textrm{gcd}(d,m))}{\textrm{gcd}(d,l,N/d,m,n,N/m)}} \, .
\end{split}
\end{equation}
We claim that the refined index is also duality invariant.

\section{Conclusions}
\label{conclusions}
We completed the table of indices of pure ${\cal N}=1$ Yang-Mills theories on a circle following \cite{Aharony:2013hda}. Using the result for the
$\mathfrak{su}(n)$ algebra,
we tallied the number of massive vacua of ${\cal N}=1^\ast$ gauge theories with $A$-type gauge algebra compactified on a circle, and demonstrated consistency with the duality of the parent ${\cal N}=4$ theory.
We refined the index by keeping track of the unbroken discrete gauge group.
It would be interesting to compute  the supersymmetric index of ${\cal N}=1^\ast$ theories
for all classical gauge algebras, for any choice of gauge
group and spectrum of line operators, and any pick of twisted boundary conditions \cite{Hanany:2001iy,Kim:2004xx}. The main ingredients will again be the supersymmetric index of pure
${\cal N}=1$ \cite{Aharony:2013hda}, the analysis of semi-classical solutions for the adjoint scalar fields, as well as the centralizer,
including discrete factors
\cite{Bourget:2015lua,Bourget:2015cza,Bourget:2015upj}.
The discrete part of the centralizers poses a hurdle -- it can be jumped.

\section*{Acknowledgments}
 We would like to acknowledge support from the grant ANR-13-BS05-0001, and from the \'Ecole Polytechnique and the \'Ecole Nationale Sup\'erieure des Mines de Paris. We are  grateful to an
 anonymous JHEP referee who has helped in improving
 our paper, and in particular suggested refining the index by classifying vacua according to their unbroken
 discrete gauge group.

\appendix

\section{\texorpdfstring{Supersymmetric Indices of Pure ${\cal N}=1$ Yang-Mills Theories}{}}
\label{indices}

In this appendix, we assemble and complete the list of the number of vacua for pure ${\cal N}=1$ super Yang-Mills theories compactified on a circle.
We follow the notation and logic of \cite{Aharony:2013hda}, to which we must refer for definitions
and details. The theories with gauge algebras $\mathfrak{b}$ and $\mathfrak{c}$
as well as $\mathfrak{d}_{N \textrm{ odd}}$ were entirely treated in \cite{Aharony:2013hda}, and we 
copied the results in Table \ref{indextable}.

The cases of Lie algebra $\mathfrak{g}_2, \mathfrak{f}_4$ and $\mathfrak{e}_8$ are trivial since the adjoint
group equals the covering group. Thus, we only need to complete the calculation for $\mathfrak{a}_{N-1}$, which is done in section \ref{pure}, and the algebras $\mathfrak{d}_{N \textrm{ even}}$, $\mathfrak{e}_6$ and $\mathfrak{e}_7$.

\begin{table}[t]
\centering
\begin{tabular}{|c|c|c|c|}
    \hline
      Algebra  & Theory  & On $\mathbb{R}^4$ & On $\mathbb{R}^3 \times S^1$ \\
        \hline
  $\mathfrak{a}_{N-1}$        & $(\mathrm{SU}(N)/\mathbb{Z}_m)_n$ & $N$ & $I^{\mathcal{N}=1}(N,m,n)$ \\
        \hline 
       \multirow{3}{*}{$\mathfrak{b}_{N \geq 2}$} & $\mathrm{Spin}(2N+1)$  &  \multirow{3}{*}{$2N-1$} & $2N-1$ \\
         & $\mathrm{SO}(2N+1)_+$  &  & $2(2N-1)$ \\
        & $\mathrm{SO}(2N+1)_-$ &   & $2N-1$ \\
        \hline 
       \multirow{3}{*}{$\mathfrak{c}_{N \geq 2}$}  & $Sp(2N)$ &   & $N+1$ \\
        & $(Sp(2N)/\mathbb{Z}_2)_+$ & $N+1$ & $\begin{cases} 2(N+1) & \textrm{ for even } N \\ \frac{3}{2}(N+1) 
       & \textrm{ for odd } N \end{cases}$ \\
        & $(Sp(2N)/\mathbb{Z}_2)_-$ &   & $\begin{cases} N+1 & \textrm{ for even } N \\ \frac{3}{2}(N+1) 
        & \textrm{ for odd } N \end{cases}$ \\
        \hline
        \multirow{3}{*}{$\mathfrak{d}_{N \geq 3}$} & $\mathrm{Spin}(2N)$ & \multirow{3}{*}{$2(N-1)$ } &  $2(N-1)$ \\
        & $\mathrm{SO}(2N)_+$ &   &  $4(N-1)$ \\
        & $\mathrm{SO}(2N)_-$ &  &  $2(N-1)$ \\
        \hline 
        $\mathfrak{d}_{N \textrm{ odd}}$ & $(\mathrm{Spin}(2N) / \mathbb{Z}_4)_n$ & \multirow{1}{*}{$2(N-1)$ } & $4(N -1)$ \\
        \hline
         \multirow{5}{*}{$\mathfrak{d}_{N \equiv 2 \, \textrm{mod} \, 4}$} & $\mathrm{Sc}(2N)_{\pm} \textrm{ and } \mathrm{Ss}(2N)_\pm$ & \multirow{5}{*}{$2(N-1)$ }  & $3(N-1)$ \\ 
        & $(\mathrm{SO}(2N)/\mathbb{Z}_2)_{\substack{ 00 \\ 00}}$ & & $5(N-1)$  \\ 
                & $(\mathrm{SO}(2N)/\mathbb{Z}_2)_{\substack{ 11 \\ 11}}$ & &  $3(N-1)$  \\ 
                & $(\mathrm{SO}(2N)/\mathbb{Z}_2)_{\substack{ 10 \\ 00}}$ & & $4(N-1)$ \\ 
             \rule[-2.2ex]{0pt}{0ex}    & $(\mathrm{SO}(2N)/\mathbb{Z}_2)_{\substack{ 01 \\ 11  }}$ & & $2(N-1)$ \\ 
                \hline 
            \multirow{6}{*}{$\mathfrak{d}_{N \equiv 0 \, \textrm{mod} \, 4}$} & $\mathrm{Sc}(2N)_+ \textrm{ and } \mathrm{Ss}(2N)_+$ & \multirow{6}{*}{$2(N-1)$ } & $4(N-1)$   \\  
                  & $\mathrm{Sc}(2N)_- \textrm{ and } \mathrm{Ss}(2N)_-$ &   & $2(N-1)$  \\ 
         & $(\mathrm{SO}(2N)/\mathbb{Z}_2)_{\substack{ 00 \\ 00}}$ & & $5(N-1)$  \\ 
                & $(\mathrm{SO}(2N)/\mathbb{Z}_2)_{\substack{ 11 \\ 11}}$ & &  $3(N-1)$  \\ 
                & $(\mathrm{SO}(2N)/\mathbb{Z}_2)_{\substack{ 01 \\ 00}}$ & & $3(N-1)$ \\ 
              \rule[-2.2ex]{0pt}{0ex}     & $(\mathrm{SO}(2N)/\mathbb{Z}_2)_{\substack{ 00 \\ 10}}$ & & $3(N-1)$ \\ 
        \hline
        \multirow{2}{*}{$\mathfrak{e}_6$} & $E_6$ & \multirow{2}{*}{$12$} & 12  \\
                       & $(E_6/\mathbb{Z}_3)_n$ & & 20  \\  
        \hline
       \multirow{2}{*}{$\mathfrak{e}_7$} & $E_7$ & \multirow{2}{*}{$18$} & 18  \\
                    & $(E_7/\mathbb{Z}_2)_n$ & & 27 \\  
        \hline            
        $\mathfrak{e}_8$ & $E_8$ & 30 & 30 \\
        \hline
        $\mathfrak{f}_4$ & $F_4$ & 9 & 9 \\
        \hline
        $\mathfrak{g}_2$ & $G_2$ & 4 & 4 \\
        \hline
\end{tabular}
    \caption{The Pure ${\cal N}=1$ Indices on a Circle. The notations are as in \cite{Aharony:2013hda}. }
    \label{indextable}
\end{table}

For the case $\mathfrak{d}_{N even}$, the hypothesis is that there is a vacuum in which purely magnetic monopoles condense. The reasoning is  standard,
based on the identification of classes of vacua transforming into each other under the global symmetry group,
and the Witten effect determining the spectrum of line operators in inequivalent vacua \cite{Aharony:2013hda}. There is a laundry list of cases to go through, and the
results are stacked in Table \ref{indextable}. We provide example calculations below.
For the exceptional groups, we divide the vacua into three groups of
$4$ for $(E_6/\mathbb{Z}_3)_n$, of which one group triples in degeneracy upon compactification, much as
for the $(\mathrm{SU}(3)/\mathbb{Z}_3)_n$ theory, and into two groups of $9$ for $(E_7/\mathbb{Z}_2)_n$, of which one group
doubles in degeneracy, similarly to what happens for $(\mathrm{SU}(2)/\mathbb{Z}_2)_n$. This leads to a total
of $3 \times 4 + 4 +4 =20$ and $2 \times 9 + 9$ vacua respectively.

To illustrate the inner workings of the added entries in Table \ref{indextable}, we provide details of their calculation
in the example of $\mathfrak{d}_{N \equiv 0 \, \textrm{mod} \, 4}$.
The algebra $\mathfrak{d}_{N \equiv 2 \, \textrm{mod} \, 4}$ can be treated similarly. 
When $N$ is even, the center of the covering group $\mathrm{Spin}(2N)$ is $\mathbb{Z}_2^S \times \mathbb{Z}_2^C$, and we will compute the index of the pure $\mathcal{N}=1$ theories with gauge group $\mathrm{Spin}(2N)/\mathbb{Z}_2^C = \mathrm{Sc}(2N)$ and $\mathrm{Spin}(2N)/(\mathbb{Z}_2^S \times \mathbb{Z}_2^C) = \mathrm{SO}(2N)/\mathbb{Z}_2$ on a circle. The $\mathrm{Ss}(2N)$ theory is 
related to the $\mathrm{Sc}(2N)$ theory by the action of the outer automorphism, and the other groups are covered in
\cite{Aharony:2013hda}.
A line operator with electric charge $(z_{e,S},z_{e,C})$ and magnetic charge $(z_{m,S},z_{m,C})$ is labelled by the element $(z_{e,S},z_{e,C};z_{m,S},z_{m,C})$ in $\mathbb{Z}_2^4$. Equivalently, we can label the line by equivalence
classes of weights $([\lambda_e] ; [\lambda_m])$ where 
\begin{equation}
[\lambda_e] , [\lambda_m] \in \{ [0] = (0,0) , [\lambda_S] = (1,0), [\lambda_C] = (0,1) , [\lambda_V] = (1,1) \} = \mathbb{Z}_2^2 \, . 
\end{equation}
From now on, we drop the square brackets. Two line operators are mutually local if and only if 
\begin{equation}
z_{e,S} z'_{m,C} - z_{m,C} z'_{e,S} \equiv z_{e,C} z'_{m,S} - z_{m,S} z'_{e,C} \quad \textrm{ mod }  2 \, , 
\end{equation}
and, combined with the fact that purely electric lines have electric charges in the weight lattice of the gauge group,
and completeness, this determines the possible line operator spectra %of the theories 
\cite{Aharony:2013hda}.

\begin{table}[t]
\centering
\begin{tabular}{|c|c|c|c|c|c|c|}
\hline
Theory & Lines $\mathscr{L}$ & \multicolumn{2}{|c|}{$p (\mathscr{C}_1 , \{ \mathscr{L} \}) $} & \multicolumn{2}{|c|}{$p (\mathscr{C}_2 , \{ \mathscr{L} \}) $} & Total \\
\hline 
\multirow{3}{*}{$\mathrm{Sc}(2N)_+$} & $(\lambda_S , 0)$ & A & & A& &  \multirow{3}{*}{$4(N-1)$}\\
& $(0, \lambda_S )$ & P & 1 & A & 1& \\
& $(\lambda_S , \lambda_S )$ & A & & P & & \\
\hline 
\multirow{3}{*}{$\mathrm{Sc}(2N)_-$} & $(\lambda_S , 0)$ & A  & & A & & \multirow{3}{*}{$2(N-1)$} \\
& $(\lambda_C, \lambda_S )$ & A  & 0 &A & 0 & \\
& $(\lambda_V , \lambda_S )$ & A & & A & & \\
\hline
\multirow{3}{*}{$(\mathrm{SO}(2N)/\mathbb{Z}_2)_{ \substack{ 00 \\ 00} }$} & $(0 , \lambda_S)$ & P & & A & & \multirow{3}{*}{$5(N-1)$} \\
& $(0, \lambda_C )$ &P & 2 & A & 0& \\
& $(0, \lambda_V )$ &P &  & A & & \\
\hline
\multirow{3}{*}{$(\mathrm{SO}(2N)/\mathbb{Z}_2)_{ \substack{11 \\ 11} }$} & $(\lambda_V , \lambda_S)$ & A & & A & & \multirow{3}{*}{ $3(N-1)$}\\
& $(\lambda_V, \lambda_C )$ & A & 1 & A & 0& \\
& $(0, \lambda_V )$ &P &  & A & & \\
\hline
\multirow{3}{*}{$(\mathrm{SO}(2N)/\mathbb{Z}_2)_{ \substack{01 \\ 00} }$} & $(\lambda_C , \lambda_S)$ &A  & & A &  & \multirow{3}{*}{$3(N-1)$}\\
& $(0, \lambda_C )$ &P & 1 & A & 0& \\
& $(\lambda_C, \lambda_V )$ &A &  & A & & \\
\hline
\multirow{3}{*}{$(\mathrm{SO}(2N)/\mathbb{Z}_2)_{ \substack{00 \\ 10} }$} & $(0 , \lambda_S)$ &P &  & A & & \multirow{3}{*}{$3(N-1)$}\\
& $(\lambda_S, \lambda_C )$ &A  & 1 & A & 0& \\
& $(\lambda_S, \lambda_V )$ &A &  & A & & \\
    \hline
\end{tabular}
\caption{The computation of the number of vacua in the compactified $\mathfrak{d}_{N \equiv 0 \, \textrm{mod} \, 4}$ theories. The letters A and P indicate that the line follows an area or perimeter law, respectively. }
\label{tableComputationVac}
\end{table}
We work under the assumption that there exists a vacuum where all the particles that condense are purely magnetic. The $2N-2$ vacua of the $\mathfrak{so}(2N)$ theories on $\mathbb{R}^4$ are divided into two classes. We have $N-1$ vacua which are related by the global discrete
symmetry group to a vacuum where particles with purely magnetic charges in the set $\mathscr{C}_1 = \{ (0 ; \alpha_i) \}_{i=1 , \dots , N}$ 
condense, where $\alpha_i$ generate a maximal rank sublattice of the root lattice. 
There are also $N-1$ vacua related to a vacuum where particles with equal electric and magnetic charges $\mathscr{C}_2= \{ (\alpha_i ; \alpha_i) \}_{i=1 , \dots , N}$ condense.
We will suppose that the  charges $\alpha_i$ correspond to simple roots of the Lie algebra. These hypotheses are
supported by the physics of ${\cal N}=2$ supersymmetric gauge theories upon breaking to ${\cal N}=1$ by mass deformation.

Now, a line operator of charge $\mathscr{L} = (\lambda_e ; \lambda_m)$ in a vacuum with condensates of charges $\mathscr{C} = \{ (\alpha_{e,i} ; \alpha_{m,i}) \}_{i=1 , \dots , N}$ on $\mathbb{R}^4$ indicates an unbroken magnetic gauge symmetry if and only if 
\begin{equation}
 \{ \lambda_e \cdot \alpha_{m,i} - \lambda_m \cdot \alpha_{e,i} \}_{i=1 , \dots , N}  = \{ 0 , \dots , 0 \}  \quad \textrm{ mod } 2  \, . 
\end{equation}
When this is the case, the line follows a perimeter law. 
In a theory with a given spectrum $\{ \mathscr{L} \}$ of line operators, the unbroken magnetic gauge group is therefore $\mathbb{Z}_2^{p (\mathscr{C} , \{ \mathscr{L} \})}$ where $p (\mathscr{C} , \{ \mathscr{L} \})$ is the rank of the lattice of charges of line operators that follow a perimeter law. 
Upon compactification, such a vacuum leads to a degeneracy equal to the cardinal number of this unbroken discrete group.
Therefore the number of vacua on $\mathbb{R}^3 \times S^1$ is 
\begin{equation}
\frac{2N-2}{2} \sum\limits_{k=1}^2 2^{p (\mathscr{C}_k , \{ \mathscr{L} \})}   \, , 
\end{equation}
where the  sum runs on the two vacua related by the transformation $\theta \rightarrow \theta+2 \pi$. The calculation is described explicitly in Table \ref{tableComputationVac}.

\section{The Details and The Arithmetic}
\label{numbers}
This appendix contains the inner workings of various technical reasonings in the bulk of the paper.
In appendix \ref{isomorphism}, we present a more detailed analysis of the unbroken gauge group argued for 
in subsection \ref{semiclassical}.
In subsections \ref{AppendixBasicSDuality}, \ref{AppendixProofGCD} and \ref{AppendixLemmas}, we provide proofs of 
useful arithmetic identities.

\subsection{The Unbroken Gauge Group}
\label{isomorphism}
Let $d$ and $m$ be divisors of $N$. In this subsection we compute in more technical detail the residual gauge group 
\begin{eqnarray}
   H_{\mathbb{C}} = S \left(  \mathrm{GL}(d)^{N/d}_\Delta \right) \cap G_{\mathbb{C}} \, . 
\end{eqnarray}
with $G_{\mathbb{C}} = \mathrm{SL}(N,\mathbb{C})/\mathbb{Z}_m$. We first note that 
\begin{equation}
    H_{\mathbb{C}} = \left[ S \left(  \mathrm{GL}(d)^{N/d}_\Delta \right) \right] / \mathbb{Z}_m \, . 
\end{equation}
Consider the map $\psi$ defined by 
\begin{eqnarray}
   \psi &:& \mathrm{SL}(d) \times \mathbb{Z}_N \rightarrow  H_{\mathbb{C}}  \\
    & & (A,l) \mapsto e^{2 \pi i l/N} \{ 1 , e^{2 \pi i /m} , \cdots , e^{2 \pi i (m-1)/m} \} \, \textrm{diag} (A , \cdots , A) \nonumber \, ,
\end{eqnarray}
where we map into elements of the set $H_{\mathbb{C}}$, which are equivalence classes.
Let us note various properties of this group homomorphism : 
\begin{itemize}
    \item It is surjective. Indeed, consider a representative $\textrm{diag} (B , \cdots , B)$ of a $\mathbb{Z}_m$ equivalence class in $ S \left(  (\mathrm{GL}(d))^{N/d}_\Delta \right)$. We have $\det B^{N/d} = 1$, so there exists $l \in \mathbb{Z}$ such that $\det B = e^{2 \pi i d l/N}$. Define $A =  e^{- 2 \pi i l/N} B$. Then $\det A = e^{- 2 \pi i d l/N} \det B = 1$. 
    \item We have $\textrm{Ker} \, \psi \cong \mathbb{Z}_{m} \times \mathbb{Z}_d$. This isomorphism is given explicitly by
    \begin{eqnarray}
    \label{actionKernel}
    \mathbb{Z}_m \times \mathbb{Z}_d & \longrightarrow & \textrm{Ker} \, \psi \subset \mathrm{SL}(d) \times \mathbb{Z}_N \nonumber \\
         (s,s') & \mapsto & \left( e^{-2\pi i \frac{s'}{d}} 1_{d\times d}, N \left(\frac{s}{m} + \frac{s'}{d} \right) \right) \, . 
    \end{eqnarray}
\end{itemize}
By the isomorphism theorem, we have 
\begin{equation}
     H_{\mathbb{C}}  \cong \frac{\mathrm{SL}(d) \times \mathbb{Z}_N}{\mathbb{Z}_m \times \mathbb{Z}_d}  \, . 
\end{equation}
We know how $\mathbb{Z}_m \times \mathbb{Z}_d$ acts on $\mathrm{SL}(d) \times \mathbb{Z}_N$ thanks to equation (\ref{actionKernel}):
\begin{equation}
(A,l) \sim \left( e^{- 2 \pi i/d } A , l+ \frac{N}{d} \right) \sim \left(A , l+ \frac{N}{m} \right) \, . 
\end{equation}
In particular, considering the least common multiple of $\frac{N}{d}$ and $\frac{N}{m}$, we can cancel the action on $l$ and obtain 
\begin{equation}
\label{actiongcd}
(A,l) \sim \left( e^{2 \pi i/ \textrm{gcd}(d,m) } A , l \right) \, , 
\end{equation}
and considering two integers $s,s'$ such that $s \frac{N}{m} + s' \frac{N}{d} = \textrm{gcd} \left( \frac{N}{m} , \frac{N}{d} \right) = \frac{N}{\textrm{lcm}(d,m)}$ gives 
\begin{equation}
\label{actionNlcm}
(A,l) \sim \left( e^{- 2 \pi i s' / d } A , l + \frac{N}{\textrm{lcm}(d,m)} \right) \, . 
\end{equation}
A  consequence of equations (\ref{actiongcd}) and (\ref{actionNlcm}) is that we have the quotient group
\begin{equation}
\frac{H_{\mathbb{C}}}{\left( \frac{\mathrm{SL}(d)}{\mathbb{Z}_{ \textrm{gcd}(d,m)}} \right)} \cong \mathbb{Z}_{\frac{N}{\textrm{lcm}(d,m)}} \, . 
\end{equation}

\subsection{\texorpdfstring{A Warm-up for  $S$-Duality}{}}
\label{AppendixBasicSDuality}

As a consequence of $S$-duality, the indices (\ref{indexSUN}) and (\ref{indexSUNZm}) of the $\mathcal{N}=1^\ast$ theories $\mathrm{SU}(N)$ and $(\mathrm{SU}(N)/\mathbb{Z}_N)_0$  have to be equal. To prove the equality, we exploit Euler's totient function $\varphi (n)$ which counts the number of integers $1 \leq k \leq n$ that are coprime to $n$. For $d$ a divisor of $n$, we have 
\begin{eqnarray}
    \varphi \left( \frac{n}{d} \right)     &=& \left| \left\{ 1 \leq k \leq n  \mid  \textrm{gcd} \left( k , n \right) = d \right\} \right| \, . 
\end{eqnarray}
We then compute 
\begin{equation}
\label{sumGauss}
    \sum\limits_{d | n} \varphi (d) =   \sum\limits_{d | n} \varphi \left( \frac{n}{d} \right) =  \sum\limits_{d | n} \sum\limits_{k=1}^n \delta_{\textrm{gcd} \left( k , n \right) = d} =  \sum\limits_{k=1}^n  \sum\limits_{d | n} \delta_{\textrm{gcd} \left( k , n \right) = d} = \sum\limits_{k=1}^n 1 = n \, , 
\end{equation}
and similarly 
\begin{equation}
\label{sumGCD}
    \sum\limits_{d | n} d \varphi \left( \frac{n}{d} \right) =  \sum\limits_{d | n} \sum\limits_{k=1}^n d \delta_{\textrm{gcd} \left( k , n \right) = d} =  \sum\limits_{k=1}^n  \sum\limits_{d | n} d  \delta_{\textrm{gcd} \left( k , n \right) = d} =  \sum\limits_{k=1}^n  \textrm{gcd} \left( k , n \right) \, .
\end{equation}
Now we can combine the two equalities to prove the desired identity: 
\begin{eqnarray}
     \sum\limits_{n | N} \sum\limits_{k=1}^n  \textrm{gcd} \left( k , n \right) &=&  \sum\limits_{n | N}  \sum\limits_{d | n} d \varphi \left( \frac{n}{d} \right) = \sum\limits_{d | N} d \sum\limits_{d | n | N} \varphi \left( \frac{n}{d} \right) \nonumber \\
     &=& \sum\limits_{d | N} d \sum\limits_{ n' | \frac{N}{d}} \varphi \left( n' \right) =  \sum\limits_{d | N} d \frac{N}{d} = N \sigma_0 (N) \, . 
\end{eqnarray}
This finishes our warm-up example for the proof of $S$-duality in the generic case in subsections \ref{genericproof} and \ref{duality}.

\subsection{A Sum of Greatest Common Divisors}
\label{AppendixProofGCD}

This subsection is devoted to the derivation of the formula (\ref{sumGCDaffine}). This is a particular case of the identity \cite{CS}
\begin{equation}
\label{identityChidambaraswamy}
    \sum\limits_{s=1}^n \textrm{gcd} \left( n , f(s) \right) = n \sum\limits_{d | n} \frac{\varphi (d)}{d} \rho (d,f)  \, , 
\end{equation}
where $f$ is any polynomial with integer coefficients, and $\rho (d,f)$ is the number of solutions of the equation $f(s)=0$ in $\mathbb{Z}_d$. Since we only need a small subset of the identities proven in \cite{CS}, we limit the proof of \cite{CS} to our simpler context. The essential ingredients are the M\"obius function identity and the exchange of two sums over divisors:
\begin{eqnarray}
\sum\limits_{s=1}^n \textrm{gcd} \left( n , f(s) \right) &=& \sum\limits_{d|n} d \sum\limits_{s=1}^n \delta \left[ \textrm{gcd} (f(s) , n) = d \right]\nonumber \\ 
&=& \sum\limits_{d|n} d \sum\limits_{s=1}^n \sum\limits_{e | \textrm{gcd} (\frac{f(s)}{d} , \frac{n}{d}) } \mu (e) \label{step1}  \\ 
&=& \sum\limits_{d|n} d \sum\limits_{e | \frac{n}{d}} \sum\limits_{s=1}^n  \mu (e) \delta \left[ de | f(s) \right]\nonumber  \\ 
&=& \sum\limits_{d|n} d \sum\limits_{e | \frac{n}{d}}  \mu (e) \rho (de ,f) \frac{n}{de} \label{step2} \nonumber \\ 
&=& \sum\limits_{d|n} d \sum\limits_{d|u|n}  \mu \left(\frac{u}{d} \right) \rho (u ,f) \frac{n}{u} \nonumber \\ 
&=& \sum\limits_{u|n} \rho (u ,f) \frac{n}{u} \sum\limits_{d|u} d   \mu \left(\frac{u}{d} \right)  \nonumber  \\ 
&=& \sum\limits_{u|n} \rho (u ,f) \frac{n}{u} \varphi (u) \label{step3} \, . 
\end{eqnarray}
We use the M\"obius function identity in steps (\ref{step1}) and (\ref{step3}) respectively. To recuperate the formula (\ref{sumGCDaffine}), we analyze the number of solutions $s$ to the Diophantine equation $as+b=0$ modulo $n$. If the equation has a solution, then there exist integers $s$ and $r$ such that $as+b=rn$, and therefore $\textrm{gcd}(a,n)|b$. If this condition is not fulfilled, the equation has no solution. If it is, then we have  $\textrm{gcd}(a,n)$ solutions. Thus, formula (\ref{identityChidambaraswamy}) implies the desired formula (\ref{sumGCDaffine}).

\subsection{Auxiliary Lemmas}
\label{AppendixLemmas}

\begin{lemma}
\label{lemma1}
Let $v$ and $d'$ be two integers such that $\mathrm{gcd} (v,d')=1$ and let $\delta '$ be any integer. Then there exists $k \in \mathbb{Z}$ such that $\mathrm{gcd} (\delta ' , v-d' k) = 1$. 
\end{lemma}

\begin{proof}
Let $P$ be the set of prime numbers that divide $\delta '$ but that do not divide $v$ nor $d'$. We now define $ k = \prod\limits_{p \in P} p$. Let $p$ be a prime that divides $\delta '$. There are three possibilities : 
\begin{itemize}
    \item If $p$ divides $v$, then $p \notin P$ so $p$ doesn't divide $k$. It doesn't divide $d'$ either because $\mathrm{gcd} (v,d')=1$. Therefore $p$ does not divide $v-d'k$, and does not divide $\mathrm{gcd} (\delta ' , v-d' k)$. 
    \item If $p$ divides $d'$, then it does not divide $v$, nor $v-d'k$, nor $\mathrm{gcd} (\delta ' , v-d' k)$. 
    \item The last possibility is that $p \in P$. But then $p$ does not divide $v$, and does not divide $\mathrm{gcd} (\delta ' , v-d' k)$. 
\end{itemize}
We see that no prime can divide $\mathrm{gcd} (\delta ' , v-d' k)$, so this number must be equal to one. 
\end{proof}

\begin{lemma}
\label{lemma2} 
Let $\delta \in \mathbb{N}^\ast$, let $d'$ and $m'$ be integers such that $(d',m')=1$ and let $n \in \mathbb{Z}$. Then there exists $a \in \mathbb{Z}$ such that 
\begin{equation}
   \mathrm{gcd} (\delta , ad'+n) = \mathrm{gcd} (\delta , n m') \, . 
\end{equation}
\end{lemma}

\begin{proof}
We use Bézout's theorem to introduce $(u,v) \in \mathbb{Z}^2$ such that $u d' + v m' = 1$. We also define $\delta ' = \frac{\delta}{(\delta , nm')}$. Note that Bézout's theorem asserts that $\mathrm{gcd} (v,d')=1$, so we can apply lemma \ref{lemma1}. Let $k$ be the integer given by this lemma and write 
\begin{equation}
    1 = (u+m'k)d' + (v-d'k)m' \, . 
\end{equation}
Then if we set $a = -n (u+m'k)$, we have 
\begin{eqnarray}
       \mathrm{gcd} \left(\delta , ad' + n \right) &=&   \mathrm{gcd}  \left(\delta , ad' + n(u+m'k)d' + (v-d'k)nm'  \right)  \nonumber \\
        &=&   \mathrm{gcd}  \left(\delta , (v-d'k)nm'  \right) \nonumber \\
        &=&  \mathrm{gcd}  \left(\delta , nm'  \right) \mathrm{gcd} \left(\delta ' , v-d'k \right) \nonumber \\
        &=& \mathrm{gcd} \left(\delta , nm'  \right) \, . 
\end{eqnarray}
\end{proof}

\begin{lemma}
\label{lemma3} 
Let $\delta \in \mathbb{N}^\ast$ and $d \in \mathbb{N}^\ast$. For $a \in \mathbb{Z}_d$, define the \emph{cyclic ordered sequence}
\begin{equation}
\label{definitionDa}
    D_a(\delta ,d) = \left[ \mathrm{gcd}(\delta , a+dt) \right]_{ t \in \mathbb{Z}_{\delta } } \, . 
\end{equation}
If the sequences $D_a (\delta ,d)$ and $ D_b (\delta ,d)$ have an element in common, then they are equal as cyclic ordered sequences. 
\end{lemma}

\begin{proof}
Let us call $\mathcal{P}(\delta,d)$ the property that we want to prove. 
Let us first show that it is sufficient to prove $\mathcal{P}(\delta,d)$ for $d$ a divisor of $\delta$. The additive group $d \mathbb{Z} / \delta \mathbb{Z}$ is isomorphic to the group $d_1 \mathbb{Z} / \delta \mathbb{Z}$ where $d_1=\mathrm{gcd} (\delta,d)$. This group isomorphism gives rise to a reshuffling of numbers inside $\mathbb{Z}_\delta$, determined by $d$ and $d_1$. The fact that the reshuffling is common between the two sequences implies that we can prove the lemma for $\mathrm{gcd} (\delta,\cdot)$ applied to the sets $\{a+ d_1 t | t \in \mathbb{Z} \}$ and $\{ b + d_1 t | t \in \mathbb{Z} \}$. 

Now we prove by induction on $\delta$ that $\forall d | \delta , \, \mathcal{P}(\delta,d) $. This is obvious for $\delta =1$. Consider then $\delta > 1$. The property $\mathcal{P}(\delta,1)$ is clear, so we choose a divisor $d>1$, and define $k= \delta / d < \delta$. Let $a,b \in \mathbb{Z}_d$ such that $D_a(\delta,d)$ and $D_b(\delta,d)$ have an element in common. With the notation $a_1 =\mathrm{gcd}(\delta,d,a)=\mathrm{gcd}(d,a)$, we have that 
\begin{equation}
\label{gcd1}
\mathrm{gcd}(\delta,a+dt) = a_1 \mathrm{gcd}\left(k \frac{d}{a_1} , \frac{a}{a_1} +\frac{d}{a_1} t \right) = a_1 \mathrm{gcd} \left(k,\frac{a}{a_1} +\frac{d}{a_1} t \right) \, . 
\end{equation}
We used the fact that $a/a_1$ is coprime with $d/a_1$ in the last equality. 

But we know that there exist $t,t'$ such that $(kd , a+dt) = (kd , b+dt')$. Then any divisor of $d$ and $a$ is a divisor of $(kd , a+dt) = (kd , b+dt')$ and therefore also a divisor of $b$. This proves that $a_1 | b_1=\mathrm{gcd}(d,b)$. Similarly, $b_1 | a_1$. Thus, $a_1=b_1$, and we deduce that 
\begin{equation}
\label{gcd2}
\mathrm{gcd}(\delta,b+dt) =  b_1 \mathrm{gcd} \left(k,\frac{b}{b_1} +\frac{d}{b_1} t \right) =  a_1 \mathrm{gcd} \left(k,\frac{b}{a_1} +\frac{d}{a_1} t \right) \, . 
\end{equation}
The sequences in the right-hand sides of (\ref{gcd1}) and (\ref{gcd2}) have an element in common, and using $\mathcal{P}(k,d/a_1)$, they are equal as cyclic sequences. Hence the same holds for the left-hand sides, proving the property $\mathcal{P}(\delta,d)$ and the lemma. 
\end{proof}

\bibliographystyle{JHEP}

\end{document}